\documentclass[10pt, conference, compsocconf]{IEEEtran}

\usepackage[cmex10]{amsmath}
\usepackage{graphicx}
\usepackage{longtable}
\usepackage{amssymb}
\usepackage{amsthm}
\usepackage{amsmath}
\usepackage{bbm}
\usepackage{dsfont}
\usepackage{epstopdf}
\usepackage{listings}
\usepackage{color}
\usepackage{arydshln}
\usepackage{pgf}
\usepackage{tikz}
\usepackage[ruled,vlined,linesnumbered]{algorithm2e}
\usetikzlibrary{arrows,shapes,snakes,automata,backgrounds,petri}
\usepackage[latin1]{inputenc}
\usepackage{wrapfig}
\usepackage{url}
\usepackage{floatrow}
\usepackage{caption}
\usepackage{subcaption}
\usepackage{enumitem}
\usepackage{gensymb}
\usepackage{times}
\setlist{nolistsep}

\usepackage[colorlinks,
            linkcolor=black,
            anchorcolor=black,
            citecolor=blue,
			urlcolor=black,
            bookmarks=true
            ]{hyperref}

\newcommand\sqra{\mathrel{]\kern-1.0ex\rightarrow}}
\newcommand\sqle{\mathrel{-\kern-0.5ex[}}
\newcommand\sqmd{\mathrel{]\kern-0.9ex-\kern-0.95ex[}}
\newcommand\sqml{\mathrel{]\kern-0.5ex-}}
\newcommand\sqmr{\mathrel{-\kern-0.5ex[}}
\newcommand\sqla{\mathrel{\leftarrow\kern-1.1ex[}}
\newcommand\sqre{\mathrel{]\kern-0.5ex-}}

\newcommand\bgal{\scriptsize\begin{align*}}
\newcommand\edal{\end{align*}\normalsize}

\newtheorem{definition}{Definition}
\newtheorem{theorem}{Theorem}
\newtheorem{lemma}{Lemma}
\newtheorem{example}{Example}

\begin{document}

\title{Stateful Security Protocol Verification}

\author{
\IEEEauthorblockN{
Li Li\IEEEauthorrefmark{1},
Jun Pang\IEEEauthorrefmark{2},
Yang Liu\IEEEauthorrefmark{3},
Jun Sun\IEEEauthorrefmark{4},
Jin Song Dong\IEEEauthorrefmark{1}
} \\
\IEEEauthorblockA{
\IEEEauthorrefmark{1}School of Computing, National University of Singapore, Singapore\\
\IEEEauthorrefmark{2}FSTC and SnT, University of Luxembourg, Luxembourg \\
\IEEEauthorrefmark{3}School of Computer Engineering, Nanyang Technological University, Singapore\\
\IEEEauthorrefmark{4}Information System Technology and Design, Singapore University of Technology and Design, Singapore
}}
\maketitle

%

\begin{abstract}
A long-standing research problem in security protocol design is how to 
efficiently verify security protocols with tamper-resistant global states. 
In this paper, we address this problem by first proposing a protocol specification framework, 
which explicitly represents protocol execution states and state transformations. 
Secondly, we develop an algorithm for verifying security properties 
by utilizing the key ingredients of the first-order reasoning for reachability analysis, 
while tracking state transformation and checking the validity of newly generated states. 
Our verification algorithm is proven to be (partially) correct, if it terminates. 
We have implemented the proposed framework and verification algorithms in a tool named SSPA, 
and evaluate it using a number of stateful security protocols. 
The experimental results show that our approach is not only feasible but also practically efficient. 
In particular, we have found a security flaw on the digital envelope protocol, 
which could not be detected by existing security protocol verifiers.

\end{abstract}


\section{Introduction} 
\label{sec:introduction}

Many widely used security protocols, e.g.,~\cite{GM99,BW00,MR08,AR10}, 
keep track of the protocol execution states. 
These protocols maintain a global state among several sessions 
and can behave differently according to the values stored in the global state. 
More importantly, the protocol's global state is tamper-resistant, 
i.e., it cannot be simply cloned, faked, or reverted. 
As the result, we cannot treat the global state as an input from the environment 
so that the protocol becomes stateless. 
In practice, such global states are usually extracted from trusted parties in protocols like 
central trustworthy databases, trusted platform modules (TPMs), etc. 


The global state poses new challenges for the existing verification techniques as discussed below. 
First, most existing verification tools, e.g., ProVerif~\cite{Bla01} and Scyther~\cite{Cre08}, 
are designed for stateless protocols. 
When they are used to verify stateful protocols, 
false alarms may be introduced in the verification results. 
For instance, when the protocol state is ignored in these tools, 
a value generated in a later global state can be used in a former global state. 
However, the execution trace is actually impractical. 
Second, stateful protocols usually have sub-processes that can be executed for infinitely many times. 
However, the state-of-the-art tools, e.g.,~\cite{Bla01,Cre08,ARR11}, cannot handle loops. 
As a consequence, only a finite number of protocol execution steps can be modeled and checked. 
Therefore, valid attacks could be missed in the verification. 
Even though some tools like Tamarin~\cite{MSCB13} can specify loops, 
the verification cannot terminate for most stateful protocols 
as they do not consider the states as tamper-resistant in the multiset rewriting rules~\cite{DLM04}. 
Third, some of the abstractions made in the existing works tend to either 
make the verification non-terminating for stateful protocols or introduce false alarms. 
For instance, fresh nonces generated in ProVerif~\cite{Bla01} are treated 
as functions to the preceded behaviors in a session 
so that the nonces with the same name could be merged under the same execution trace. 
On one hand, if a stateful protocol receives some data before generating any nonce in its session, 
the nonce generated in one session can be received 
before the same nonce is generated in a different session. 
According to the abstraction method, the nonce becomes a function applied to itself, 
which could lead to infinite function applications. 
Thus the verification cannot terminate. 
On the other hand, if a nonce is generated without performing any session-specific behavior, 
then the nonce will be the same for multiple sessions. 
The query of asking whether the nonce for a particular session can be deduced may give false alarms, 
because the nonce that can be deduced is actually coming from another session. 
As these nonces are merged, they cannot be differentiated in the verification process. 

To address the above identified challenges for verifying stateful security protocols,
we first propose a protocol specification framework (see Section~\ref{sec:specification})
that explicitly models the protocol execution state as tamper-resistant. 
We specify how states are used in the protocol as well as how states are transferred. 
As a result, stateful protocols can be modeled in our framework in an intuitive way. 
The protocol specification is introduced with a motivating example 
of the digital envelope protocol~\cite{AR10}. 
Second, a solving algorithm is developed to verify stateful protocols. 
During solving, we apply a pre-order to the states and converge the states into a valid state trace. 
The secrecy property checked in this work is then formulated into a reachability problem. 
The partial correctness of our method is formally 
defined in Section~\ref{sub:query} and proved in Section~\ref{sec:verification}. 
However, as the security protocol verification problem is undecidable in general~\cite{MSDL99}, 
our algorithm does not guarantee the termination. 
The experiments show that our method can terminate for  
many stateful security protocols used in the real world. 
Third, we develop a tool named SSPA (Stateful Security Protocol Analyzer) based on our approach. 
Several stateful protocols including the digital envelop protocol
and the Bitlocker protocol~\cite{BL11} have been analyzed using SSPA. 
The experiment results show that our method can both find security flaws and give proofs efficiently. 
Particularly, we have found a security flaw in the digital envelope protocol which has not been identified before. 

\medskip\noindent
{\bf Structure of the paper.}
Related works are discussed in Section~\ref{sec:related} and 
a motivating example is given in Section~\ref{sec:example}. 
In Section~\ref{sec:specification}, we present our protocol specification framework
and describe how to specify cryptographic primitives, protocols and queries. 
In Section~\ref{sec:verification}, we show how the verification algorithm works 
and prove its partial correctness. 
We show the implementation details and the experiment results in Section~\ref{sec:experiments}. 
Finally, we conclude the paper with some discussions in Section~\ref{sec:discussions}. 


\section{Related Works} 
\label{sec:related}

M{\"o}dersheim developed a verification framework that works with global states~\cite{Mod10}. 
His framework extends the IF language with sets 
and abstracts the names based on its Set-Membership. 
According to~\cite{Mod10}, this method works well for several protocols. 
However, its applicability in general is unclear
since sets should be explicitly identified for the protocols 
and no general solution for identifying the set is given in the paper. 
Guttman extended the strand space with mutable states 
to deal with stateful protocols~\cite{Gut09,Gut12}, but there is no tool support for his approach. 
Our approach presented in this paper is different from theirs, as the protocol specification 
does not need to be changed in our framework and we provide automatic tool support. 

StatVerif is introduced by Arapinis et al.~\cite{ARR11} to verify protocol with explicit states. 
It extends the process calculus of ProVerif~\cite{Bla01} with stateful operational semantics 
and translates the resulting model into Horn clauses. 
ProVerif is then used as an engine to perform verification. 
Comparing with their method that can only work with a finite number of (global) states, 
our approach is more general and works for protocols with infinite states. 

In~\cite{DKRS11}, Delaune et al.\ modeled TPMs with Horn clauses and have verified three protocols using ProVerif. 
However, the specifications need to be adapted according to the different protocols under study. 
For instance, an additional parameter is added into the global state 
when it is used for the digital envelope protocol (DEP)~\cite{AR10} to prevent false attacks. 
More importantly, they also modified the specification of the DEP
in a way that false negatives can happen (attacks are missing) 
comparing with the original DEP proposed in~\cite{AR10}. 
This is because their method does not work for infinite steps of the stateful protocols. 
Specifically, they have constrained the protocol so that its second phase is not repeatable. 
More discussions on the DEP can be found in Section~\ref{sec:example}. 
Notice that all of the previous methods can only work with protocols with finite steps. 
while this is not the case with our approach. 


\section{Motivating Example} 
\label{sec:example}

\begin{figure}[t]
	\begin{center}
	\includegraphics[scale=0.48]{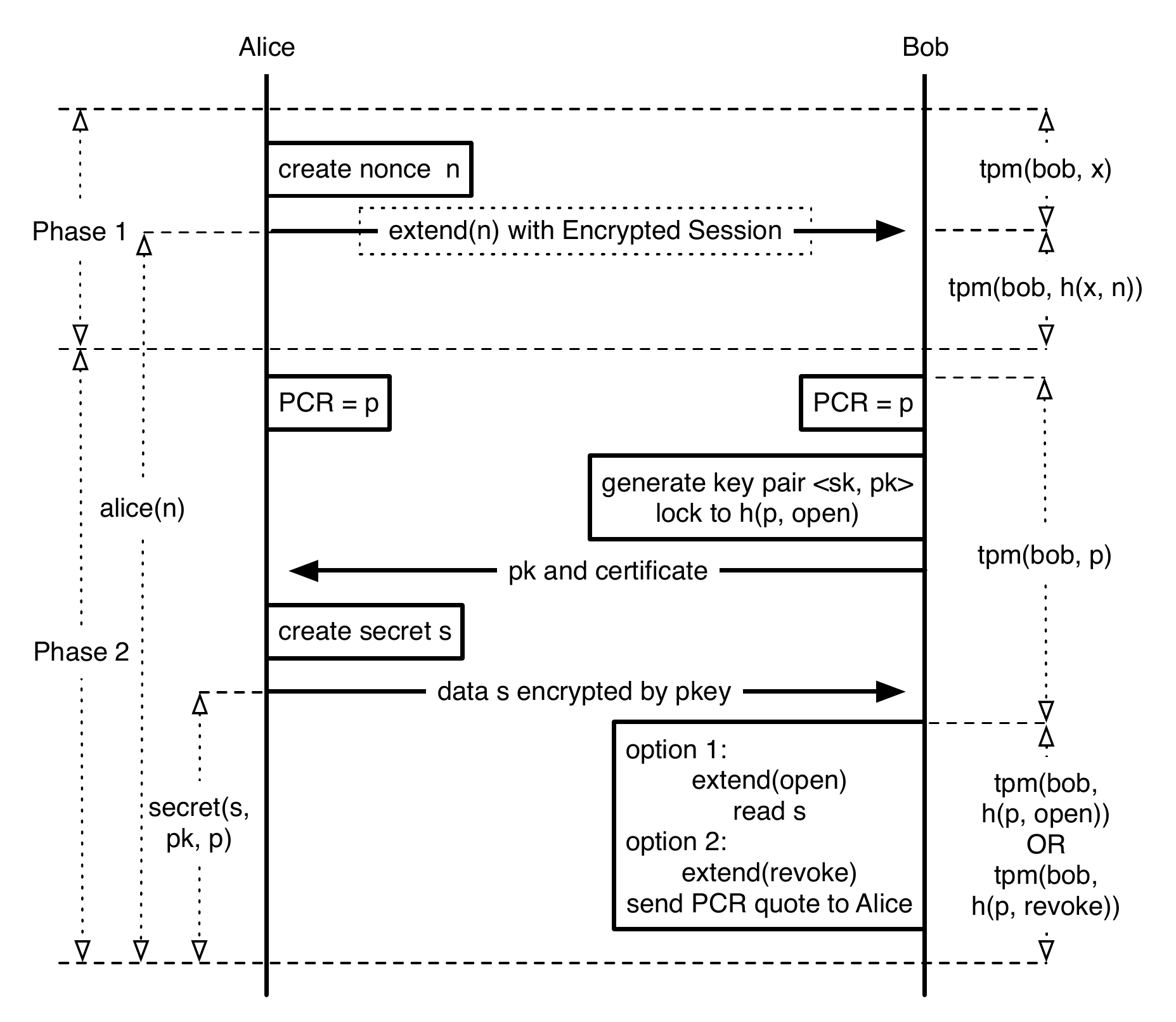}
	\caption{The digital envelope protocol (DEP)}
	\label{fig:envelope}
	\end{center}
\end{figure}

We introduce the digital envelope protocol (DEP)~\cite{AR10} in this section as a motivating example. 
Before going into the details of the protocol, 
we give a brief introduction on the trusted platform module (TPM)~\cite{TPM} used in the protocol first. 

TPM is an embedded cryptographic device proposed to give higher level security guarantees 
than those can be offered by software alone. 
Every TPM has several tamper-resistant platform configuration registers (PCRs) 
that maintain the current state of the TPM. 
The values stored in the PCRs can only be \emph{extended}. 
One possible implementation of extending a PCR $p$ with a value $n$ could be 
$\mathit{extend}(n)\{ p = \mathit{h}(p, n) \}$,
where $\mathit{h}$ is a one-way hash function applied to the concatenation of $p$ and $n$. 
Hence, the extending actions are irreversible 
unless the TPM reboot is allowed (the PCRs are reset to the default value $b$) 
and the previous extending actions are replayed in an identical order. 
TPM provides several APIs to help the key management, 
including key generation, key usage, etc., under PCR measurement. 

TPMs use several types of keys, including 
the attestation identity keys (AIKs) and the warp keys. 
The AIK represents the identity of the TPM in the protocol and can be used for signing. 
In order to differentiate the TPMs, we assume every TPM has a unique AIK. 
However, this assumption does not prevent the adversary from using multiple AIK values 
as he could initiate multiple TPMs. 
The warp keys form a tree structure rooted under the permanent loaded storage root key (SRK). 
We usually use two kinds of warp keys in the TMP, i.e., the binding keys and the storage keys. 
Data can be encrypted with the binding public key remotely, 
or can be sealed with the loaded storage key in the TPM. 
Typically, the TPM supports the following operations. 
\begin{itemize}
    \item \emph{Extend}. Extend the PCR value $p$ by any value $n$ to a new PCR value $h(p, n)$. 
    \item \emph{Read}. Read the current PCR value from the TPM. 
    \item \emph{Quote}. Certify the current PCR value. 
    \item \emph{CreateWrapKey}. Generate a warp key under a loaded parent key 
    and bind it to a specific PCR value. 
    The new key is not yet loaded into the TPM but stored in a key blob, 
    which is a storage place for holding the key. 
    \item \emph{LoadKey2}. Load the key into TPM by providing the key blob and its parent key. 
    \item \emph{CertifyKey}. Certify a loaded key. 
    \item \emph{UnBind}. Decrypt the data with a loaded binding key. 
    The PCR value for the key should be matched. 
    \item \emph{Seal}. Encrypt the data with a storage key. 
    The PCR value for the key should be matched 
    and the encrypted data can be sealed to a particular PCR value. 
    \item \emph{UnSeal}. Decrypt the data with the loaded storage key. 
    The PCR value of the seal key, the PCR value of the sealed storage 
    and the current PCR value are required to be the same. 
\end{itemize}
As the storage key and seal/unseal operation are not used in the DEP, 
we omit their specification in the following discussions. 

By using TPMs, the DEP allows an agent \emph{Alice} 
to provide a digital secret $s$ to another agent \emph{Bob} 
in a way that \emph{Bob} can either access $s$ without any further help from \emph{Alice}, 
or revoke his right to access the secret $s$ so that he can prove his revocation. 
This protocol consists of two phases as shown in Figure~\ref{fig:envelope}. 
In the first phase, \emph{Alice} generates a secret nonce $n$ 
and uses it to extend a given PCR in \emph{Bob}'s TPM with an encrypted session. 
The transport session is then closed. 
Since the nonce $n$ is secret, \emph{Bob} cannot re-enter the current state of the TPM 
if he makes any changes to the given PCR. 
In the second phase, \emph{Alice} and \emph{Bob} read the value of the given PCR as $p$ 
and \emph{Bob} creates a binding key pair $\langle sk, pk \rangle$ locked to the PCR value $\mathit{h(p, open)}$ 
and sends the key certification to \emph{Alice}, 
where $\mathit{open}$ is an agreed constant in the protocol. 
This means the generated binding key can be used 
only if the value $\mathit{open}$ is first extended to the PCR of value $p$.
After checking the correctness of the certification, 
\emph{Alice} encrypts the data $s$ with her public key $pk$ and sends it back to \emph{Bob}. 
Later, \emph{Bob} can either open the digital envelope by extending the PCR with $\mathit{open}$ 
or revoke his right to open the envelope by extending another pre-agreed constant $\mathit{revoke}$. 
If \emph{Bob} revokes his right, the quote of PCR value $\mathit{h(p, revoke)}$ 
can be used to prove \emph{Bob}'s revoke action. 
The protocol is illustrated in Figure~\ref{fig:envelope}.

\begin{figure}[t]
    \centering
        \includegraphics[scale=0.45]{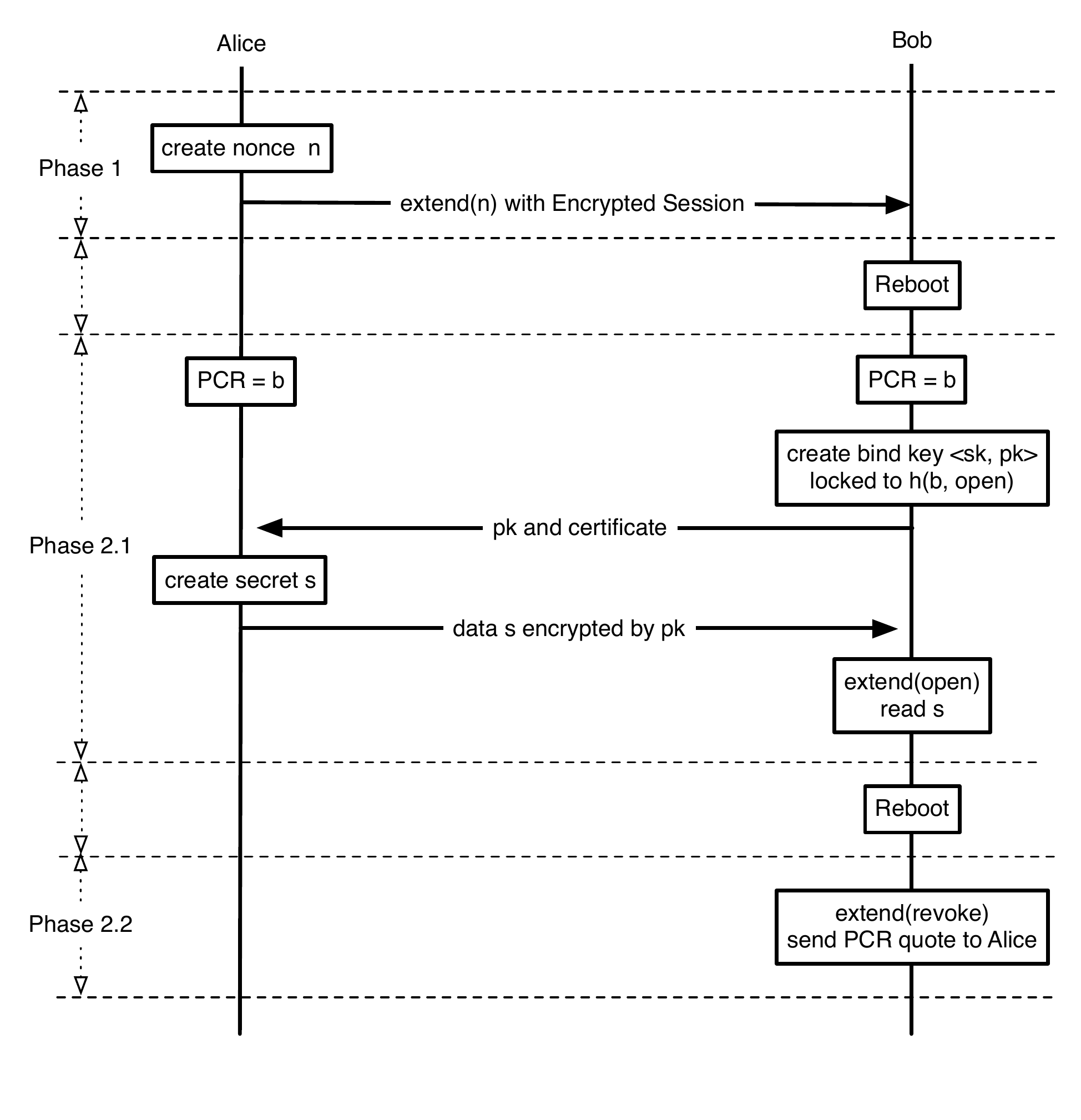}
    \caption{An attack on the DEP}
    \label{fig:attack}
\end{figure}

In fact, through our approach and the implemented tool,
we have found a cold-boot attack for this DEP 
when the TPM reboot is allowed. 
According to the DEP proposed in~\cite{AR10}, 
the authors only mentioned that \emph{Bob} may lose his ability to open the envelope 
or to prove his revoke action if the TPM reboot is allowed. 
To the best of our knowledge, this attack has not been described before.
We present the attack scenario in Figure~\ref{fig:attack}. 
When the TPM reboot is allowed, \emph{Bob} can reboot his TPM 
immediately after the first phase. 
As a consequence, the secret nonce $n$ extended to the given PCR is lost. 
When \emph{Alice} checks the PCR value in the beginning of the second phase, 
she actually reads a PCR value that is unrelated to her previous extend action. 
Hence, \emph{Bob} can re-enter the current TPM state by simply performing TPM reboot again. 
This attack is caused by the fact that 
the PCR value in the second phase can be unrelated to the PCR value in the first phase. 
On the other hand, if the TPM reboot is not allowed, 
the secret nonce $n$ could never get lost. 
So \emph{Alice} can conduct the second protocol phase for multiple times 
and the claimed properties of the DEP are always preserved. 
In this way, if the TPM is maintained by a trusted server and remotely controlled by 
both \emph{Alice} and \emph{Bob} without the right to reboot TPM, this protocol is secure. 

This protocol was previously verified in~\cite{DKRS11}. 
However, the modifications made in~\cite{DKRS11} to the original DEP prevent the authors from detecting the attack. 
In the modified version~\cite{DKRS11}, \emph{Bob} always does the TPM reboot before the first phase 
and \emph{Alice} assumes that the PCR is $h(b, n)$ 
without actually reading the value in the beginning of the second phase. 
As a result, TPM reboot can never happen before the second phase. 
The reason why they need to make such modifications is because  
ProVerif, which is used in their verification, can only model finite protocol steps. 
Unfortunately, this makes it impossible to find the attack as described in Figure~\ref{fig:attack}. 
On the contrary, in this work, we provide a framework 
where protocols like this can be modeled faithfully and verified automatically. 


\section{Protocol Specification} 
\label{sec:specification}

In this section, we describe our specification framework 
for modeling (stateful) protocols, crypto primitives and queries 
as a set of first order logic rules with the protocol execution states explicitly maintained. 
There are two categories of rules that can be specified in our approach, 
i.e., state consistent rules and state transferring rules. 
The state consistent rules specify the knowledge deductions, 
while the state transferring rules describe the state transitions. 
Since the protocol global state is tamper-resistant, 
we assume that it can only be changed by the state transition rules. 
The adversary model we consider in this work is the standard active attacker, 
who can intercept all communications, compute new messages and send any messages that he can obtain or compute. 
For instance, he can use all the public available functions including encryptions, decryptions and etc. 
He can also ask the legitimate protocol participants to take part in the protocol. 
That is, every rule specified in the framework describes a logic capability of the adversary. 
Our goal is to check whether he can deduce a target fact or not. 

\subsection{Framework Overview} 
\label{sub:overview}

In our framework, every entity and device in the protocol 
is treated as an object when it is tamper-resistant. 
Every object have an object global state with a unique identity. 
The protocol global state then consists of several object global states. 
For simplicity, we name object global state after state for short, 
and call protocol global state as protocol state in the remaining of this paper. 
For instance, every TPM has a state $\mathit{tpm(aik, p)}$ 
which records the AIK value $\mathit{aik}$ and the PCR value $p$. 
The AIK value uniquely identifies the TPM. 
Initially, the protocol state of DEP is $\{ \mathit{tpm(bob, p)} \}$, 
where $\mathit{bob}$ stands for the AIK constant for \emph{Bob}'s TPM. 
After the first phase of the DEP, \emph{Alice} enters a state $\mathit{alice(n)}$ 
where $n$ is the secret value that she extends to Bob's TPM. 
When \emph{Alice} obtains the state, she can initiate the second phase of the protocol. 
Because \emph{Alice} could start several sessions to \emph{Bob}'s TPM, 
we treat $n$ as the identity of \emph{Alice}'s state. 
When \emph{Alice} extends the nonce $n$ to \emph{Bob}'s TPM, 
the protocol state becomes $\{ \mathit{tpm(bob, h(p, n))}, \mathit{alice(n)} \}$. 
A protocol state can contain several TPM states with different AIK values. 

The states of the same object should be ordered in a timeline of protocol execution, forming a state trace. 
For instance, the following sequence of four states is a legitimate TPM state trace in the DEP. \vspace{2mm}
\begin{enumerate}
    \item $\mathit{tpm(bob, i)}$
    \item $\mathit{tpm(bob, h(i, n))}$
    \item $\mathit{tpm(bob, h(h(i, n), x))}$
    \item $\mathit{tpm(bob, h(h(h(i, n), x), revoke))}$
\end{enumerate}
\vspace{2mm}\noindent
The first state is the initial state. 
Then, in the first phase of the DEP, 
Alice extends a secret nonce $n$ into \emph{Bob}'s TPM (the second state). 
Later, \emph{Bob} extends a value $x$ into his TPM for other purposes (the third state) 
and the second phase of the DEP begins. 
At the beginning, \emph{Alice} and \emph{Bob} record the PCR value as $h(h(i, n), x)$. 
When \emph{Bob} receives \emph{Alice}'s sealed secret, 
\emph{Bob} extends the pre-agreed constant $\mathit{revoke}$ 
to revoke his right of opening the envelope (the fourth state). 
In most protocols, one state can be used for multiple times. 
For instance, in the above example, \emph{Bob} needs to use the third TPM state for several times 
to generate key, load key, generate certifications and etc. 
As these states are actually the same, 
we need to identify them as one state when they are used in different places. 
On the other hand, the first state used in the protocol is precedent to the third state. 
Thus, we should also identify how states are updated, 
namely the transformation an old state to a new state. 

The protocol rules specified in our framework are of the form $H : M \sqle S : O \sqra V$. 
$H$ is a set of premises such as 
the terms that the adversary should know and the events that the protocol should engage. 
$S$ is a set of states. 
Both of $H$ and $S$ must be satisfied so that the rule is applicable. 
For example, when the adversary wants to load a key into the TPM, 
the adversary should know its parent key and obtain the TPM state with matched PCR value. 
$V$ is the conclusion of the rule with two types of values. 
One type of conclusion is a fact. 
Take the TPM loading key as an example, its conclusion is a fact that 
the adversary can get the loaded key in the TPM. 
The other type of conclusion represents 
how the states are transferred from old ones to new ones. 
As the states in our framework are attached to the objects, 
the conclusion consists of pairs of old state and new state for the same object, 
denoting that state is converted from one to another. 
In TPM extending operation, the conclusion is one pair of states 
$\langle tpm(aik, p), tpm(aik, h(p, n)) \rangle$ 
in which the PCR value in the second state is extended. 
$M$ and $O$ help us to organize the correspondences between facts and states. 
$M$ maps the facts to the states indicating that the facts should be known at which states. 
$O$ is the orderings of the states generated from the knowledge deduction. 
For instance, when a fact $f$ required by a rule $R$ can be provided as the conclusion of another rule $R'$, 
we can compose these two rules together to remove the requirement of $f$. 
Since the $f$ is provided by $R'$ and used in $R$, 
the states mapped by $f$ in $R$ are required later than requirement of the states in $R'$. 
The orderings are specified in the verification process 
to make sure that the state trace is practical for the protocol. 
We name the rule as \emph{state consistent rule} when $V$ is a fact 
and call the rule as \emph{state transferring rule} when $V$ is a set of state conversions. 

In addition, we use \emph{events} and \emph{states} to distinguish the protocol sessions. 
The \emph{events} are engaged in the rule predicates to indicate the generation of fresh nonces. 
Since fresh nonces are random numbers, we assume their values can uniquely identify the events. 
Whenever the nonces generated in different events have the same value, 
these events should be merged. 
On the other hand, the \emph{states} are used to describe the objects or entities presented in the protocol. 
Basically, we use states to differentiate the different phases of the objects. 
As we do not bound the number of \emph{events} and \emph{states}, 
the verification is conducted for an infinite number of sessions. 


\subsection{Term Syntax} 
\label{sub:term}

\begin{table}[t]
\small
	\begin{center}
    \begin{tabular}{ | l | l r | }
        \hline \hline 
        Type				& Expression  & \\
        \hline 
		Data($x$) 			& $*\mathit{n}$ & (key name)\\
							& $\mathit{n}$ & (message name)\\
        \hline 
		Declaration($D$)	& $s(x_1, x_2, \ldots, x_n)$ & (state type)\\
							& $e(x_1, x_2, \ldots, x_n)$ & (event type)\\
		\hline 
        Term($t$)			& $f(t_1, t_2, \ldots, t_n)$ & (function) \\
							& $a[]$ & (name) \\
							& $[n]$ & (nonce) \\
							& $|g|$ & (configuration) \\
							& $v$ & (variable) \\
		\hline
        State($s$)  		& $s(t_1, t_2, \ldots, t_n)$ & (state) \\ 
        \hline 
        Fact($f$)      	    & $k(t)$ & (knowledge) \\
                        	& $e(t_1, t_2, \ldots, t_2)$ & (event) \\
        \hline 
		Conversion($c$)	    & $\langle s, s' \rangle$ & (state conversion) \\
        \hline 
        Rule($R$)			& \multicolumn{2}{l|}{$f_1, f_2, \ldots, f_n : M 
                                \sqle s_1, s_2, \ldots, s_m : O \sqra$} \\
                            &   $f$ & (state consistent rule) \\
							& \multicolumn{2}{l|}{$f_1, f_2, \ldots, f_n : M 
                                \sqle s_1, s_2, \ldots, s_m : O \sqra$} \\ 
							&   $c_1, c_2, \ldots, c_k$ & (state transferring rule) \\
		\hline
        Accessibility($A$) & $s(t_1, t_2, \ldots, t_n)$ & (state instance) \\
		\hline
    \end{tabular}
	\end{center}
    \caption{Rule syntax hierarchy}
    \label{tab:hierarchy}
\end{table}

We adopt the syntax in Table~\ref{tab:hierarchy} to model the protocols. 
Before using an event or a state in the rules, 
we need to declare it with a unique identity. 
For the nonce generation event, the pair of the event name and the fresh nonce is 
the key\footnote{Note that it is different from a cryptographic key.} and 
we can merge two events if they have the same key. 
While for states, the pair of the object name and the object identity is the key and 
states with the same key should be ordered, describing certain phases of the same object. 

Rules are used to specify the protocol execution and adversary capabilities. 
They have the hierarchy structure as follows. 
\emph{Terms} could be defined as \emph{functions}, \emph{names}, 
\emph{nonces}, \emph{configurations} or \emph{variables}. 
\emph{Functions} can be applied to a sequence of \emph{terms}; 
\emph{names} are globally shared constants; 
\emph{nonces} are freshly generated values in the sessions; 
\emph{configurations} are values pre-existed in the states; 
and \emph{variables} are memory locations for holding the \emph{terms}. 

\emph{States} describe object stages in the protocol by maintaining a set of terms. 
If two states $s$ and $s'$ have the same key, 
they are describing the same object, denoted as $s \sim s'$. 
The operator $\sim$ is an equivalence relation that can partition a state set into several disjoint subsets. 
When a mutable value is encoded in the \emph{state}, we name it as \emph{configuration}. 
It is different from variables because its value is decided by the environment, 
while the value of a variable is decided by the assignment to the variable. 
In other words, \emph{configuration} is pre-existed while \emph{variable} is post-assigned. 

A \emph{fact} can be the engagement of an \emph{event}, 
or it means that a term $t$ is known to the adversary denoted as $k(t)$. 
We define \emph{mapping} as a pair of \emph{fact} $f$ and \emph{state} $s$ 
denoted by $\langle f, s \rangle$, representing that $f$ is true at state $s$. 
Additionally, we define state \emph{ordering} by applying the binary operator $\le$ over state pairs: 
$s \le s'$, i.e., $s$ should be a state used no later than $s'$. 
The state set is a preorder set over $\le$, 
and each $\sim$ partition is a partially ordered set over $\le$. 
The derivation of mappings and orderings are discussed in Section~\ref{sec:verification}. 

A \emph{conversion} $c$ is a pair of states $\langle s, s' \rangle$ 
which stands for the transformation from an old state $s$ to a new state $s'$. 
We call $s$ as the \emph{pre-state} of $c$ denoted as $\mathit{pre}(c)$ 
and name $s'$ as the \emph{post-state} of $c$ denoted as $\mathit{post}(c)$. 
For a set of \emph{conversions} $C$, we have $\mathit{pre}(C) = \{ \mathit{pre}(c) | c \in C \}$ 
and $\mathit{post}(C) = \{ \mathit{post}(c) | c \in C \}$. 
There are two kinds of \emph{rules} that can be specified 
in our framework as shown in the Table~\ref{tab:hierarchy}. 
The \emph{state consistent rule} means if $f_1, f_2, \ldots, f_n$ are true 
under the protocol state $s_1, s_2, \ldots, s_m$ satisfying the mappings $M$ and the orderings $O$, 
$f$ is also true under the same state. 
For the \emph{state transferring rule}, it means if $f_1, f_2, \ldots, f_n$ are true 
under the protocol state $\{ s_1, s_2, \ldots, s_m \} \cup pre(C)$ satisfying the mappings $M$ and the orderings $O$, 
the protocol state can be transferred into $\{ s_1, s_2, \ldots, s_m \} \cup post(C)$ where $C = c_1, c_2, \ldots, c_k$. 

Assume $H$ is a fact set, $S$ and $S'$ are two state sets, 
we define $H \times S = \{ \langle f, s \rangle | f \in H, s \in S \}$ 
and $S \times S' = \{ s \le s' | s \in S, s' \in S' \}$. 
Given a rule $H : M \sqle S : O \sqra V$ directly specified from the protocol, 
the predicates $H$ should be given at the exact states and all the states should be presented at the same time. 
So the default value of $M$ is $H \times S$,  
and the default value of $O$ is $S \times S$. 
In the remaining of the paper, we omit them in the protocol specification. 



\subsection{Rule Modeling} 
\label{sub:modeling}

In the following, we illustrate how to specify stateful protocols in our approach
by using the DEP described in Section~\ref{sec:example} as a running example. 
In the following protocol specification, we assume that 
both of the first phase and the second phase could be conducted for infinitely many times. 
We assume that all of the values extended to \emph{Bob}'s TPM in the first phase 
and the secrets bound to the public key in the second phase are freshly generated nonces. 
So we can differentiate the sessions and values used in the sessions during the verification. 

In order to clearly illustrate the modeling strategy employed in our approach, 
we describe the basic functionalities of the TPM along with the rules. 
Notice that our approach is not limited to the applications of TPM,
but potentially other stateful security protocols. \vspace{2mm}

\subsubsection{Declarations} 
\label{ssub:declaration}

Before specifying the protocol, 
we need to declare the events and the states that are used in the rules and queries. 

There are three nonce generation events in the DEP. 
The $\mathit{genkey(*sk, aik, p, pcr)}$ event models that a new binding key $sk$ is generated in the TPM. 
In addition to the fresh key $sk$, 
the $\mathit{genkey}$ event also specifies the AIK value $aik$ 
and the PCR value $p$ of the TPM when the key is generated. 
Moreover, the $\mathit{pcr}$ in the $\mathit{genkey}$ event models the PCR value that $sk$ is bound to. 
The $\mathit{init(*n, p)}$ event is emitted 
when \emph{Alice} extends the nonce $n$ to \emph{Bob}'s TPM of the PCR value $p$. 
The $\mathit{gensrt(*s, p, pkey)}$ event is engaged 
when \emph{Alice} creates the secret $s$ for a new session of the second phase 
after receiving a key certification of $pkey$ issued from \emph{Bob}'s TPM with the PCR value $p$. 

In terms of the protocol states, 
\emph{Alice} enters the state $\mathit{alice(*n)}$ 
after she extends the secret nonce $n$ to \emph{Bob}'s TPM. 
\emph{Alice} also maintains the state $\mathit{secret(*s, p, pkey)}$ 
when she decides to share the secret value $s$ over \emph{Bob}'s TPM with the PCR value $p$. 
The $pkey$ is a public key generated from \emph{Bob}'s TPM, locked to PCR $h(p, open[])$. 
Beside, every TPM has a state of $\mathit{tpm(*aik, p)}$ 
in which the TPM is identified by the AIK value $aik$ and it has the PCR value $p$. \vspace{2mm}


\subsubsection{State Consistent Rules} 
\label{ssub:state_preserved_rules}

The rules in the first category preserves the protocol execution state. 
However, they can be applied only if the protocol is in some specific states. 
Most of the rules related to the TPM fall into this category. 

\smallskip\noindent
{\bf Stateless Rules.} 
Some stateless operations are allowed in stateful protocols 
such as encryption, decryption, concatenation and etc. 
For instance, public key generation and the binding operation of the TPM can be modeled as 
\normalsize
\begin{eqnarray}
&\mathit{k(skey)} \sqle~\sqra \mathit{k(pk(skey))} \label{rule:pk}  \\
&\mathit{k(mess)}, \mathit{k(pkey)} \sqle~\sqra \mathit{k(aenc(mess, pkey))} \label{rule:aenc} 
\end{eqnarray}
\normalsize
where the state set is empty in these rules. 
Rule (\ref{rule:pk}) means that if the adversary knows a term $\mathit{skey}$, 
he could treat it as a private key and compute its corresponding public key $pk(\mathit{skey})$. 
Rule (\ref{rule:aenc}) models the binding operation happened outside of the TPM, 
which means if the adversary knows a message $\mathit{mess}$ and a binding public key $\mathit{pkey}$, 
he could encrypt $\mathit{mess}$ by $\mathit{pkey}$ 
and get the asymmetric encryption $\mathit{aenc}(\mathit{mess}, \mathit{pkey})$. 
As stateless protocols can be considered a special case of stateful protocols, 
our verification framework also works for stateless protocols. 
Other two stateless rules in the DEP model the fact that 
the agreed constant values $\mathit{open}$ and $\mathit{revoke}$ are known publicly. 
\normalsize
\begin{eqnarray}
&\sqle~\sqra \mathit{k(revoke[])} \label{rule:revoke} \\
&\sqle~\sqra \mathit{k(open[])} \label{rule:open}
\end{eqnarray}
\normalsize

\smallskip\noindent
{\bf Data Fetch Rules.} 
Another category of the state consistent rules contains the data fetch rules. 
They model the fact that some data used in the protocol can be fetched 
directly from the protocol state without other information. 
In the DEP, the adversary has control over the TPM.  
First of all, he can use the the storage root key (SRK) to encrypt any messages. 
In addition, he can ask the TPM for its PCR value and its PCR quote without providing any information.  
To specify a general case of the TPM, 
the AIK value is not fixed to \emph{Bob}'s TPM. 
\normalsize
\begin{eqnarray}
& \sqle \mathit{tpm(|aik|, |p|)} \sqra k(\mathit{srk(|aik|)}) \label{rule:srk} \\
& \sqle \mathit{tpm(|aik|, |p|)} \sqra k(\mathit{|p|}) \label{rule:pcr} \\
& \sqle \mathit{tpm(|aik|, |p|)} \sqra k(\mathit{pcrcert(|aik|, |p|)}) \label{rule:pcrcert}
\end{eqnarray}
\normalsize
As $\mathit{srk(|aik|)}$ represents the SRK itself rather than its value,  
rule (\ref{rule:srk}) means that the adversary has access to the SRK. 
Rule (\ref{rule:pcr}) and (\ref{rule:pcrcert}) stand for 
getting the PCR value and the PCR quote, respectively. 
PCR quote is a certification issued from the TPM 
that can be used to prove its PCR value. 

\smallskip\noindent
{\bf Data Processing Rules.}
The third category of the state consistent rules contains data processing rules, 
which process data based on the presented information and the protocol state. 
As we have illustrated in Section~\ref{sec:example}, 
the keys used in the TPM are well protected and strictly controlled. 
In the TPM, keys can only be generated under a parent key, 
and the generated key can be bound to a specific PCR value so that 
it can be used only if the given PCR is of that value. 
In the DEP, for the sake of simplicity, 
we assume all the new keys are generated from the SRK. 
Additionally, all the new keys are bound to a specific PCR value 
as it is the case for the  protocol. 
Notice that our technique does not restrict us from 
specifying the complete TPM. 
\normalsize
\begin{align}
&k(\mathit{pcr}), k(\mathit{srk(|aik|)}), \mathit{genkey([sk], |aik|, |p|, pcr)} \notag \\ 
	&~~\sqle \mathit{tpm(|aik|, |p|)} \sqra \notag \\
    &~~k(\mathit{\langle pk([sk]), blob(|aik|, [sk], srk(|aik|), pcr) \rangle}) \label{rule:blob} \\
&k(\mathit{blob(|aik|, sk, pakey, pcr)}), k(\mathit{pakey}) \notag \\
	&~~\sqle \mathit{tpm(|aik|, |p|)} \sqra 
    k(\mathit{pcrkey(|aik|, sk, pcr)}) \label{rule:pcrkey} 
\end{align}
\begin{align}
&k(\mathit{pcrkey(|aik|, sk, pcr)}) \notag \\ 
	&~~\sqle \mathit{tpm(|aik|, |p|)} \sqra 
    k(\mathit{keycert(|aik|, pk(sk), pcr)}) \label{rule:keycert} \\
&k(\mathit{aenc(data, pk(sk))}), k(\mathit{pcrkey(|aik|, sk, |p|)}) \notag \\
	&~~\sqle \mathit{tpm(|aik|, |p|)} \sqra 
    k(\mathit{data}) \label{rule:unbind} 
\end{align}
\normalsize
Rule (\ref{rule:blob}) specifies that a new session key $sk$ can be generated 
in the TPM identified by $aik$ with PCR value $p$. 
In addition, the new key is bound to the PCR value $\mathit{pcr}$ 
so that it can only be used when the PCR is of that value. 
As can be seen from rule (\ref{rule:blob}), 
we need to specify the target PCR value for the key and provide the SRK as well. 
In addition, all of the related information should be encoded into the key generation event 
so that it can be used to identify the key generation behavior. 
Initially, the generated key is not loaded into the TPM but stored in a key blob. 
So rule (\ref{rule:pcrkey}) models the key loading operation by providing the key blob and its parent key. 
When the key is loaded, the TPM can issue key certification as illustrated in rule (\ref{rule:keycert}). 
Rule (\ref{rule:unbind}) describes the bound data can be decrypted with the corresponding loaded key. 
More importantly, the PCR value specified in the key should be matched with the current PCR. 

When \emph{Alice} receives key certification from \emph{Bob} and she has already finished the first phase, 
she generates a secret $[s]$, encrypts it with the public key $pkey$ and sends it to \emph{Bob}. 
\normalsize
\begin{align}
    \mathit{gensrt([s], p, pkey)}, &\mathit{k(keycert(bob[], pkey, h(p, open[])))} \notag \\
    	\sqle \mathit{alice(|n|)} &\sqra k(\mathit{aenc([s], |pkey|)}) \label{rule:secret}
\end{align}
\normalsize\vspace{1mm}


\subsubsection{State Transferring Rules} 
\label{ssub:state_transition_rules}

The state transferring rules change the protocol's global state. 
The PCR value extending action is modeled as follows. 
\normalsize
\begin{align}
	k(n) \sqle~\sqra \langle \mathit{tpm(|aik|, |p|)}, \mathit{tpm(|aik|, h(|p|, n))} \rangle \label{rule:extend}
\end{align}
\normalsize
Rule (\ref{rule:extend}) means that if the adversary knows a value $n$, 
he could extend the given PCR in the TPM by $n$. 
The second state transition rule models the first phase for \emph{Alice}. 
\normalsize
\begin{align}
\mathit{init([n], |p|)} &\sqle~\sqra \langle , \mathit{alice([n])} \rangle, \notag \\
	&\langle \mathit{tpm(bob[], |p|)}, \mathit{tpm(bob[], h(|p|, [n]))} \rangle \label{rule:duplication}
\end{align}
\normalsize
The constant $bob[]$ is the AIK value of \emph{Bob}'s TPM. 
\emph{Alice} enters a state called \emph{alice} 
after \emph{Alice} confirms that the nonce $n$ is extended to \emph{Bob}'s TPM. 
Meanwhile, the nonce $n$ is extended to \emph{Bob}'s TPM as described in the protocol.  
After the \emph{alice} state is presented, 
\emph{Alice} could repeatedly conduct the second phase of the protocol for infinitely many times. 

The optional rule (\ref{rule:reboot}) below specifies the reboot behavior of the TPM. 
\normalsize
\begin{align}
	\sqle~\sqra \langle \mathit{tpm(|aik|, |p|)}, \mathit{tpm(|aik|, boot[])} \rangle \label{rule:reboot}
\end{align}
\normalsize
In this work, we prove that the digital envelope protocol is secure 
when the TPM reboot is disallowed. 
We also show that this protocol is subject to attack otherwise.  



\subsection{Accessibility} 
\label{sub:accessibility}

Besides the rules, we also need to specify the object accessibilities for the adversary. 
The accessibility describes the objects the adversary have access to. 
So given a state in a rule, 
we can decide whether the states can be accessed by the adversary or not. 
For instance, in the DEP, the adversary can access \emph{Bob}'s TPM, 
and he can use additional TPMs to process messages if necessary. 
\normalsize
\begin{align*}
    \mathit{access}~&\mathit{tpm(bob[], |p|)} \\
    \mathit{access}~&\mathit{tpm(|aik|, |p|)}
\end{align*}
\normalsize
We match the state patterns by substituting the terms in the states. 
We discuss more details about accessibility and pattern matching in Section~\ref{sub:searching}. 


\subsection{Query} 
\label{sub:query}

In this paper, we focus on reachability properties such as secrecy. 
For instance, we want to ensure that \emph{Bob} cannot open the secret $s$ 
as well as obtain the proof for his revoke action 
$\mathit{certpcr(bob[], h(p, \mathit{revoke}[]))}$ at the same time 
for any iteration $\mathit{secret(s, p, pkey)}$ in the DEP. 
If he can, it means that \emph{Bob} can cheat in the protocol. 
We add supplementary rules to represent 
whether the adversary has the ability to obtain certain terms as events, 
such that we could simply check if those events are reachable or not. 

We need to first add another state transferring rule when we want to check reachability. 
This rule models that \emph{Alice} has indeed accepted the certification of the key. 
\normalsize
\begin{align}
\mathit{gensrt([s], p, pkey)}, &\mathit{k(keycert(bob[], pkey, h(p, open[])))} \notag \\
	\sqle \mathit{alice(|n|)} &\sqra \mathit{\langle , secret([s], p, pkey) \rangle} \label{rule:iteration}
\end{align}
\normalsize

The queries are generally state consistent rules, 
but they have event conclusions. 
In the DEP, we are interested in the reachability properties as follows. 
\normalsize
\begin{align}
    &\mathit{gensrt([s], |p|, |pkey|)}, \mathit{k([s])} \notag \\
        &~~~~~~~~\sqle \mathit{secret([s], |p|, |pkey|)} \sqra \mathit{opened()} \label{query:opened} \\
    &\mathit{gensrt([s], |p|, |pkey|)}, \mathit{k(pcrcert(bob[], h(|p|, \mathit{revoke}[])))} \notag \\
        &~~~~~~~~\sqle \mathit{secret([s], |p|, |pkey|)} \sqra \mathit{revoked()} \label{query:revoked} \\
    &\mathit{gensrt([s], |p|, |pkey|)}, \mathit{k(pcrcert(bob[], h(|p|, \mathit{revoke}[])))} \notag \\
        &~~~~~~~~, \mathit{k([s])}\sqle \mathit{secret([s], |p|, |pkey|)} \sqra \mathit{attack()} \label{query:attack} 
\end{align}
\normalsize
The first query (rule \ref{query:opened}) means that \emph{Bob} can open the envelope and extract the nonce $[s]$. 
Similarly, the second query (rule \ref{query:revoked}) means that the PCR quote can be issued from the TPM 
if \emph{Bob} chooses to revoke the right of opening the envelope. 
The third query (rule \ref{query:attack}), the most interesting one, checks 
whether \emph{Bob} can get the value of the nonce $[s]$ 
as well as the proof for his revoke action from his TPM at the same time. 
As can be seen, we can name the events differently 
and check several queries at the same time. 

Because verification for security protocol is generally undecidable, 
our algorithm cannot guarantee termination. 
Hence we define correctness under the condition of termination (partial correctness) as follows. 
In Section~\ref{sec:verification}, we present our verification algorithm 
on reachability checking, together with its partial correctness proofs. 

\begin{definition}[Partial Correctness] 
	A verification algorithm is \textbf{partially sound} if and only if 
    the target event is reachable when the algorithm can terminate and claim that the event is reachable. 
    It is \textbf{partially complete} if and only if 
    the target event is unreachable when the algorithm can terminate and claim that the event is unreachable. 
\end{definition}



\section{Verification Algorithm} 
\label{sec:verification}

After a protocol is correctly specified (as illustrated in Section~\ref{sec:specification}), 
we present how to verify the protocol in details in this section. 
During the verification, we divide our algorithm into two phases. 
The first phase is targeted at constructing a knowledge searching base 
by \emph{knowledge forward composition} and \emph{state backward transformation}. 
Based on the knowledge base, we could then perform query searching 
to find valid attacks in the second phase. 

In order to verify security protocols, 
the verification algorithm needs to consider all possible behaviors of the adversary. 
Because the adversary adopted in this work can generate new names dynamically at runtime, 
the verification process cannot be conducted in a straightforward manner. 
To guide the attack searching procedure so that it can terminate, 
we adopt a similar strategy as proposed in~\cite{Bla01} that applies to the Horn theory. 

Our algorithm can be briefly described as follows. 
Recall that a rule of the form $H : M \sqle S : O \sqra V$ says that 
the $V$ is true when all the predicates in $H$ are satisfied and all the states $S$ are presented 
under the restrictions of state mappings $M$ and orderings $O$. 
On one hand, if a predicate in a rule is not yet satisfied, 
we try to use a state consistent rule's conclusion to fulfill it by \emph{rule composition}. 
However, if the predicate is a singleton, that is a fact of the form $k(v)$ where $v$ is a variable, 
and the value of $v$ is not related to other facts in the rule, 
the singleton could be automatically fulfilled as the adversary assumed in our paper can generate new names. 
Additionally, events are not unifiable in our framework 
as the events in the predicates and the conclusions are different. 
Thus we reserve a set of facts $\mathcal{N}$ from unifying with other facts. 
In this work, $\mathcal{N}$ consists of \emph{events} and \emph{singletons}. 
On the other hand, 
if several states are presented in a rule, 
some of the states should be the latest ones that are presented when the conclusion is given, 
while others are the outdated states. 
Thus, we identify the latest states 
and deduce them to their previous states with the help of \emph{rule transformation}. 
By performing the \emph{rule composition} and \emph{rule transformation} iteratively, 
once the fixed-point can be reached for the knowledge base, 
the query can then be answered directly from the rules in the knowledge base. 

\subsection{Knowledge Base Construction} 
\label{sub:knowledge_base_construction}

In this section, we compose existing rules to generate new rules 
until the fixed point of the searching knowledge base is reached. 
Basically, when we compose two rules together, 
the term encoded in the conclusion of the first rule should be unifiable with 
the term in a predicate of the second rule. 
We use the most general unifier to unify the terms. 

\begin{definition}[Most General Unifier]
\label{def:unification}
    If $\sigma$ is a substitution for both terms $t_1$ and $t_2$ so that $\sigma t_1 = \sigma t_2$, 
    we say $t_1$ and $t_2$ are unifiable and $\sigma$ is a unifier for $t_1$ and $t_2$. 
    If $t_1$ and $t_2$ are unifiable, the most general unifier for $t_1$ and $t_2$ is a unifier $\sigma$,  
    where for all unifiers $\sigma'$ of $t_1$ and $t_2$ 
    there exists a substitution $\sigma''$ such that $\sigma' = \sigma''\sigma$. 
\end{definition}

The unification of the facts is defined if and only if 
their predicate names are matched and the corresponding terms in the facts can be unified. 
According to Section~\ref{sec:specification}, 
we have two kinds of rules in our framework, 
i.e., \emph{state consistent rules} and \emph{state transferring rules}. 
\emph{State consistent rules} have a fact as conclusion, 
so given an unsatisfied predicate in a rule, 
we can compose the \emph{state consistent rule} to it to provide the predicate. 
The rule composition is formally defined as follows. 
\begin{definition}[Rule Composition]
\label{def:composition} 
	Let $R = H : M \sqle S : O \sqra f$ be a state consistent rule 
    and $R' = H' : M' \sqle S' : O' \sqra V$ be either a state consistent rule or a state transferring rule. 
    Assume there exists $f_0 \in H'$ such that $f$ and $f_0$ are unifiable with the most general unifier $\sigma$. 
    Given $S_0 = \{ s_0 | \langle f_0, s_0 \rangle \in M' \}$, 
    the rule composition of $R$ with $R'$ on the fact $f_0$ is defined as 
    \normalsize 
    \begin{align*} 
        R \circ_{f_0} R' &= \sigma(H \cup (H' - \{ f_0 \})) : \sigma(M \cup M') \\
        &\sqle \sigma (S \cup S' : O \oplus O' \oplus S \times S_0) \sqra \sigma V. 
    \end{align*} 
    \normalsize 
\end{definition} 

\begin{example}
For instance, given two simplified rules as follows. 
We omit the mappings and orderings when they are trivial 
and use special characters (e.g., $\spadesuit$, $\blacklozenge$) 
to indicate the facts and states in the mappings and orderings. 
\begin{align*}
    \mathit{gensrt([s], |p|, pkey)}& \sqle \mathit{tpm(bob[], h(|p|, open[]))}^{\spadesuit} \sqra \mathit{k([s])} \\
    \mathit{gensrt([s], |p|, pkey)}&, \mathit{[s]}^{\blacklozenge} : \{ \langle \blacklozenge, \clubsuit \rangle \} \\
        \sqle \mathit{tpm}&\mathit{(bob[], h(|p|, revoke[]))}^{\clubsuit} \sqra \mathit{attack()}
\end{align*}
The first rule means that the secret $s$ can be revealed 
when \emph{Bob}'s TPM has the PCR value $\mathit{h(p, open[])}$. 
The second rule means if \emph{Bob}'s TPM has the PCR value $\mathit{h(p, revoke[])}$ 
and the secret $s$ is revealed (the envelope is opened), 
we have found an attack. 
Their rule composition on the fact $f_0 = \mathit{k([s])}$ is 
\begin{align}
    \mathit{gensrt}(&[s], |p|, \mathit{pkey}) \sqle \mathit{tpm(bob[], h(|p|, open[]))}^{\spadesuit}, \notag \\
            \mathit{tpm}&\mathit{(bob[], h(|p|, revoke[]))}^{\clubsuit} : \spadesuit \le \clubsuit 
            \sqra \mathit{attack()} \label{rule:comp}
\end{align}
which means that $\mathit{open[]}$ should be extended to \emph{Bob}'s TPM 
before $\mathit{revoke[]}$ is extended. 
This is apparent because the last state of \emph{Bob}'s TPM, 
according to the rules, should have $\mathit{revoke[]}$ extended. 
\end{example}

Given a state consistent rule with a conclusion $f$, 
it specifies that we can obtain $f$ if its predicates are provided and the states form a valid state trace. 
Furthermore, some of the states are the latest states when the conclusion is given. 
Among the latest states, the latest state transformation is taken on some of them. 
If we can identify those latest states for the latest state transformation, 
we then can deduce their precedent states using the corresponding state transferring rule. 
We define $S_0$ as the cover set of $S$
if $s_0 \in S_0, s \in S, s_0 \le s$ then $s \in S_0$. 
Assume $c$ is a conversion and $\mathit{post}(c)$ is unifiable with a state $s$ under $\sigma$, 
we define the join operator $c \bowtie_{\sigma} s = \sigma \mathit{pre}(c)$. 
Besides, we define $[s]^S$ as the $\sim$ partition of $s$ in the state set $S$. 
The state transformation is then defined as follows. 
\begin{definition}[State Transformation]
\label{def:transformation} 
	Let $R = H : M \sqle S : O \sqra C$ be a state transferring rule 
    and $R' = H' : M' \sqle S' : O' \sqra f$ be a state consistent rule. 
    Assume there exists a unifier $\sigma'$ and an injective function $m : C \rightarrow \mathds{P}(S')$ 
    such that $\cup_{c \in C} \sigma' m(c)$ is a cover set of 
    $\cup_{c \in C} [\sigma' \mathit{post}(c)]^{\sigma' S}$ and 
    $\forall c \in C, \forall s \in m(c), c \bowtie_{\sigma'} s$ is defined. 
    Let $\sigma$ be the most general unifier of $\sigma'$ and 
    $S_n = \sigma S' - \mathit{post}(\sigma C)$, 
    the state transformation of applying $R$ to $R'$ on $m$ is defined as 
    \normalsize 
    \begin{align*} 
        &R \bowtie_{m} R' = \sigma (H \cup H') : \sigma (M \cup M') \sqle 
            \sigma S \cup S_n \cup \mathit{pre}(\sigma C) \\
            &: \sigma O \oplus \sigma O' 
            \oplus \mathit{pre}(\sigma C) \times \mathit{pre}(\sigma C) 
            \oplus (\oplus_{c \in C} (([\sigma pre(c)]^{\sigma S} \\
            &~~~~- \sigma m(c)) \times \sigma pre(C) \oplus \sigma pre(C) \times \sigma m(c))) \sqra \sigma f.  
    \end{align*} 
    \normalsize 
\end{definition} 

\begin{example}
For instance, if the PCR value extending rule (\ref{rule:extend}) 
is used for transferring the states in rule (\ref{rule:comp}), 
we first enumerate the state cover set of rule (\ref{rule:comp}) 
as $\{ \clubsuit \}, \{ \spadesuit, \clubsuit \}$. 
Because the states of $\{ \spadesuit, \clubsuit \}$ cannot be unified, 
we have only one valid rule after the state transformation. 
\begin{align*}
    \mathit{gensrt}([s], |p|, \mathit{pkey}), \mathit{k(revoke[])} 
    \sqle &\mathit{tpm}\mathit{(bob[], |p|)}^{\heartsuit}, \notag \\
            \mathit{tpm(bob[], h(|p|, open[]))}^{\spadesuit} &: \spadesuit \le \heartsuit 
            \sqra \mathit{attack()} \label{rule:comp}
\end{align*}
Since the new generated rule has an unsatisfied predicate that is not in $\mathcal{N}$, 
the verification algorithm continues. 
However, when TPM reboot is disallowed, 
these two states remained in the rule can never be unified to one state, 
so the \emph{attack} event cannot be reached. 
The detailed discussions are available in the reachability analysis.
\end{example}

The adversary can generate new names. 
If a singleton predicate is not related to other facts in a rule, 
the adversary could generate a random fact and use it as the singleton predicate 
so that it can be removed from the predicates. 
In addition, given two events with the same key in the predicates, 
they should be unified and merged. 
Furthermore, for any two states $s \sim s'$ and $s \le s' \land s' \le s \in O$, 
they should be merged because clearly they are the same state. 
Meanwhile, any mappings and orderings 
related to the non-existing facts and states should be removed as well. 
\begin{definition}[Rule Validation] 
\label{def:validation} 
	Let $R = H : M \sqle S : O \sqra V$ be a rule. 
    We define a rule as valid if and only if there exists a unifier $\sigma'$ such that 
    any event in $H$ under the same key is unifiable with $\sigma'$. 
    Let $\sigma$ be the most general unifier of $\sigma'$, 
    The rule validation of $R$ is defined as 
    \normalsize 
    \begin{align*} 
        R \Downarrow = &\mathit{clear}(\mathit{merge}(\sigma H : \mathit{rm}(\sigma M))) \\
            &\sqle \mathit{elim}(\sigma S : \mathit{rm}(\sigma O)) \sqra \sigma V 
    \end{align*} 
    \normalsize 
	The function $\mathit{merge}$ merges duplicated expressions; 
    the function $\mathit{clear}$ removes any singleton 
    in which the variable does not appear in other facts in the rule; 
    the function $\mathit{elim}$ eliminates any isolated states and those related orderings; 
    and the function $\mathit{rm}$ removes the mappings and orderings related to no longer existed facts and states. 
\end{definition} 

When a new rule is composed from existing ones, 
we need to make sure it is not redundant. 
Suppose two rules $R$ and $R'$ can make the same conclusion, 
while (1) $R$ requires less predicates, mappings and orderings than $R'$ and (2) $R$ is no less general than $R'$,
$R'$ should be implicated by $R$. 
The joint operator `$\cdot$' between mapping $M$ and ordering $O$ is defined as 
\[
    M \cdot O = \{ \langle f, s \rangle | \langle f, s' \rangle \in M \land s' \le s \in O \}.
\]
We then define rule implication as follows. 
\begin{definition}[Rule Implication]
\label{def:implication} 
	Let $R = H : M \sqle S : O \sqra V$ and $R' = H' : M' \sqle S' : O' \sqra V'$ be two rules. 
	We define $R$ implies $R'$ denoted as $R \Rightarrow R'$ if and only if 
	$\exists \sigma, \sigma V = V' \land \sigma H \subseteq H' \land \sigma (M \cdot O) \subseteq (M' \cdot O')
		\land \sigma S \subseteq S' \land \sigma O \subseteq O'$. 
\end{definition}

The knowledge base construction algorithm is shown in Algorithm~\ref{alg:base},
where we use $\mathcal{B}_{init}$ to denote the initial set of rules as specified
and use $\mathcal{B}$ to denote the knowledge base constructed by the algorithm. 
In the following discussions, we will use $\mathcal{B}$ and $\mathcal{B}_{init}$
directly assuming they are clear from the context.

In the $\mathit{add}$ procedure (Line 1 to Line 6), we use rule implication to ensure that 
redundancies will not be introduced into the knowledge base. 
The main procedure, starting at Line 7, first adds all the initial rules into the knowledge base (Line 8 to Line 11), 
then it composes and transforms the rules until a fixed point is reached. 
We discuss the rule composition and the state transformation separately as follows. 

For the rule composition (Line 13 to Line 20), when rules can be composed in an unlimited method, 
infinitely many composite rules can be generated, 
which we shall prevent.
For instance, we can compose the rule (\ref{rule:pk}) to itself 
by treating the public key as a valid private key 
and the composite rule becomes $\mathit{k(skey)} \sqle~\sqra \mathit{k(pk(pk(skey)))}$, 
which could then be composed to the rule (\ref{rule:pk}) again. 
Furthermore, as mentioned previously, 
singleton predicates that are not related to other facts in the rule can be removed, 
thus it is unnecessary to compose two rules on a singleton fact. 
As the rules cannot compose on events, when two rules are composed in our algorithm, 
we need to ensure that they can be composed on a fact $f_0$ such that $f_0 \not \in \mathcal{N}$. 
Moreover, when two rules are composed in the form of $R \circ_{f_0} R'$ 
and $R$ has predicates which are not contained in $\mathcal{N}$, 
we should fulfill those predicates first. 
Thus we ensure that $R$'s predicates are all in $\mathcal{N}$. 

For the state transformation (Line 21 to Line 28), as we deduce the states in a backward manner, 
we should make sure that the states we transferred in the rule are latest, 
and the target event is presented in the rule conclusion. 
In addition, its predicates should be all contained in $\mathcal{N}$, 
resulting from the same reason mentioned previously.

Finally, we select a subset of the rules. 
Their predicates should only be singletons and events 
as rules with unfulfilled predicates cannot be used to conduct attacks directly. 
Their conclusion should be an event 
because these rules are the only interesting rules to us. 
$\mathcal{B}_{v}$ is introduced in Line~\ref{alg:base:bv} to help the explanation of the proof for Theorem~\ref{thm:base}. 

\begin{algorithm}[t]
\small
    \SetAlgoLined
    \SetKwInOut{Input}{Input}
    \SetKwInOut{Output}{Output}
    \Input{$\mathcal{B}_{init}$ - initial rules}
    \Output{$\mathcal{B}$ - knowledge base}
    \SetKwFunction{add}{$\mathit{add}$}
    \SetKwProg{proc}{Procedure}{}{}
    \SetKwProg{algo}{Algorithm}{}{}
    \proc{\add{R, rules}}{ \label{alg:base:add}
        \For{$R_b \in \mathit{rules}$}{
            \lIf{$R_b \Rightarrow R$}{\Return $\mathit{rules}$}
            \lIf{$R \Rightarrow R_b$}{$\mathit{rules} = \mathit{rules} - \{ R_b \}$}
        }
        \Return $\{ R \} \cup \mathit{rules}$; \\
    }
    \algo{}{
        $\mathit{rules} = \emptyset$; \\
        \For{$R \in \mathcal{B}_{init}$}{ \label{alg:base:bi}
            $\mathit{rules} = \mathit{add}(R, \mathit{rules})$; \\
        }
        \Repeat{fix-point is reached}{ \label{alg:base:fix:b}
                \textbf{Case 1.} \\
                Select a state consistent rule $R = H \sqle S : O \sqra f$ \\
                and a general rule $R' = H' \sqle S' : O' \sqra V$ \\
                from $\mathit{rules}$ such that \\
                1. $\forall p \in H: p \in \mathcal{N}$; \\
                2. $\exists f_0: f_0 \not \in \mathcal{N}$; \\
                3. $(R \circ_{f_0} R') \Downarrow$ is valid; \\
                $\mathit{rules} = \mathit{add}((R \circ_{f_0} R') \Downarrow, \mathit{rules})$; \\
                \textbf{Case 2.} \\
                Select a state transferring rule $R = H \sqle S : O \sqra C$ \\
                and a general rule $R' = H' \sqle S' : O' \sqra f$ \\
                from $\mathit{rules}$ such that \\
                1. $\forall p \in H \cup H': p \in \mathcal{N}$; \\
                2. $f$ is an event; \\
                3. $\exists m, (R \bowtie_m R') \Downarrow$ is valid; \\
                $\mathit{rules} = \mathit{add}((R \bowtie_m R') \Downarrow, \mathit{rules})$; \\
        } \label{alg:guided:fix:e}
        $\mathcal{B}_v = \mathit{rules}$;\label{alg:base:bv} \\ 
        \Return $\mathcal{B}= \{ R \in \mathit{rules} | \forall p \in \mathit{predicates}(R), p \in \mathcal{N} 
            \land \mathit{conclusion}(R)~\mathit{is}~\mathit{an}~\mathit{event} \}$;\\
    }
    \caption{Knowledge Base Construction}
    \label{alg:base}
\end{algorithm}

\begin{figure}[t]
    \centering
    \begin{subfigure}[b]{0.46\textwidth}
        \includegraphics[width=\textwidth]{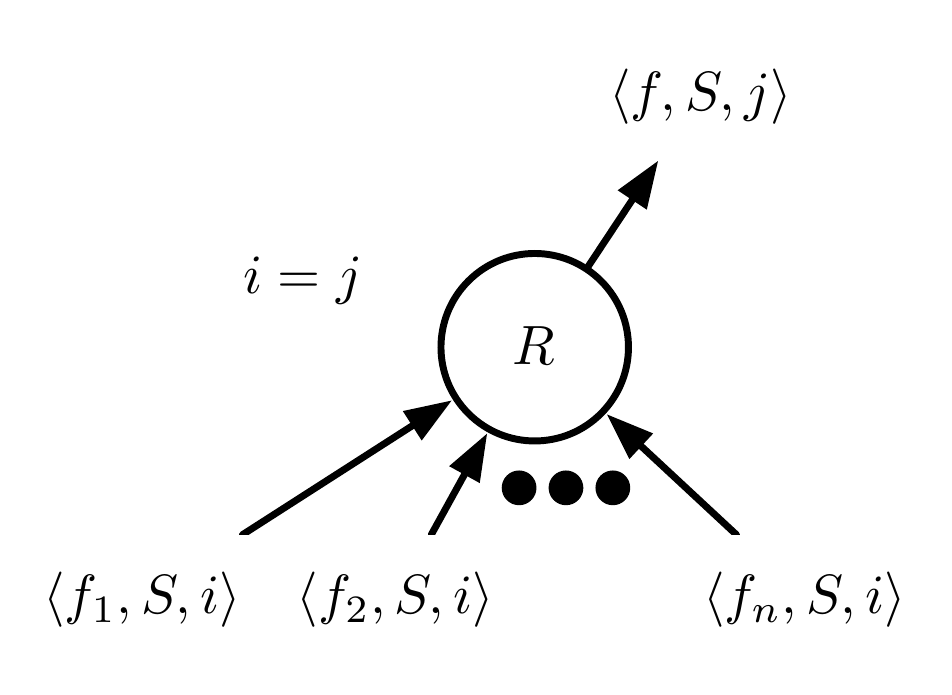}
        \caption{State Consistent Rule}
        \label{fig:tree_consistent}
    \end{subfigure}
    \begin{subfigure}[b]{0.50\textwidth}
        \includegraphics[width=\textwidth]{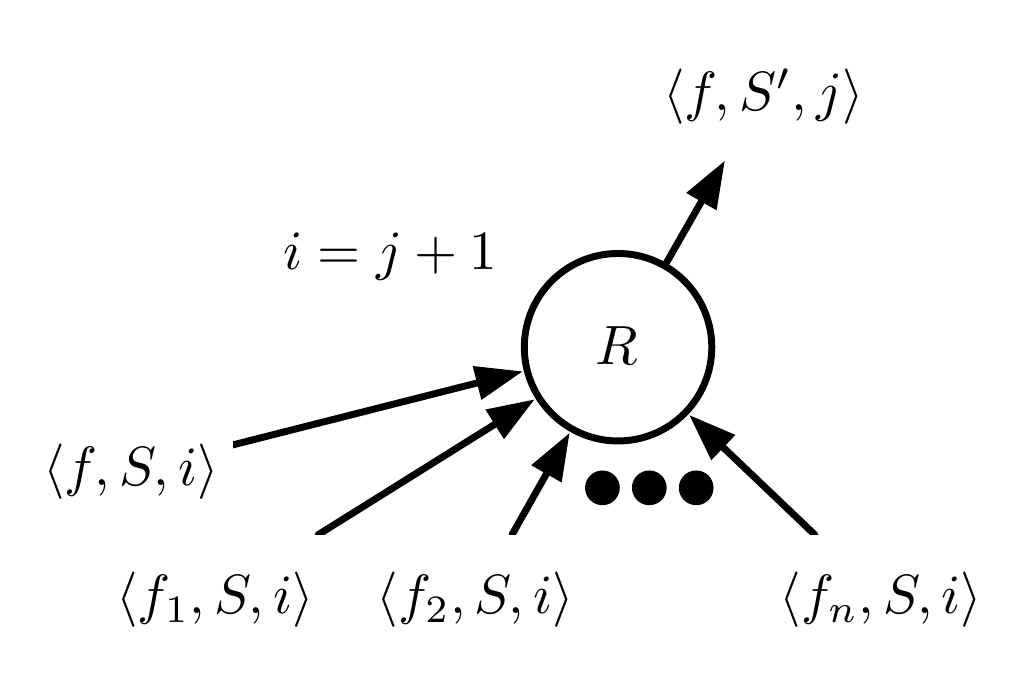}
        \caption{State Transferring Rule}
        \label{fig:tree_transfer}
    \end{subfigure}
    \caption{Rule in derivation tree}
    \label{fig:rule_in_tree}
\end{figure}

Previously, we have reformulated our verification problem as
reachability analysis of events (see Section~\ref{sub:query}). 
Whenever an event is derivable from the initial rules $\mathcal{B}_{init}$, 
there must exist a derivation tree for that event defined as follows. 
\begin{definition}[Derivation Tree]
\label{def:derivation}
    Let $\mathcal{B}$ be a set of closed rules and $e$ be an event, 
    where the closed rule is a rule with its conclusion initiated by its predicates and states. 
    $e$ can be derived from $\mathcal{B}$ if and only if 
    there exists a finite derivation tree defined as follows. 
	\begin{enumerate}
		\item Every edge in the tree is labeled by a fact $f$, a state set $S$ and an index $i$, 
        and $\forall s, s' \in S$ we have $s \not \sim s'$. 
		\item Every node is labeled by a rule in $\mathcal{B}$. 
		\item Suppose the node is labeled by a state consistent rule 
        as shown in Figure~\ref{fig:tree_consistent}, 
        then we have $R \Rightarrow H : M \sqle S : O \sqra f$ in which 
        $H = f_1, \ldots, f_n$, $M = H \times S$, $O = S \times S$ 
        and the indexes labeled on the outgoing edge and incoming edges are the same. 
		\item On the other hand, if the node is labeled by a state transferring rule 
        as shown in Figure~\ref{fig:tree_transfer}, 
        there exists $C$ such that $R \Rightarrow H : M \sqle S_0 : O \sqra C$ 
        in which $H = f_1, \ldots, f_n$, $S_0 = S - pre(C) = S' - post(C)$, 
        $M = H \times S_0$, $O = S_0 \times S_0$ 
        and the indexes labeled on the incoming edges equal to the index labeled on the outgoing edge plus $1$. 
		\item The outgoing edge of the root is labeled by the event $e$ and the index $1$. 
		\item The incoming edges of the leaves are only labeled by facts in $\mathcal{N}$ with the same index. 
		\item The edges with the same index have the same state. 
	\end{enumerate}
\end{definition}
In the tree, every node is labeled by a rule in $\mathcal{B}_{init}$ 
to represent how the knowledge is deduced. 
Additionally, 
we label the edges with states to indicate when the knowledge deduction rule is applied 
and how the state transferring rule affects the states.
Furthermore, we also label every edge with an index
to group the knowledge under the same state together as well as 
to denote the valid trace of state transferring, 
which eases the proof of Theorem~\ref{thm:base}. 

The Lemma~\ref{lem:combine} demonstrates how to replace two directly connected nodes in the derivation tree 
with one node labeled by a composite rule with the same state and the same index. 
\begin{lemma}\label{lem:combine}
    If $R_o \circ_{f} R'_o$ is defined, $R_t \Rightarrow R_o$ and $R'_t \Rightarrow R'_o$, 
    then either there exists $f'$ such that $R_t \circ_{f'} R'_t$ is defined 
    and $R_t \circ_{f'} R_t' \Rightarrow R_o \circ_{f} R'_o$, 
    or $R_t' \Rightarrow R_o \circ_{f} R'_o$. 
\end{lemma}
\begin{proof}
    Let $R_o = H_o : M_o \sqle S_o : O_o \sqra f_o$, 
    $R'_o = H'_o : M'_o \sqle S'_o : O'_o \sqra V_o$, 
    $R_t = H_t : M_t \sqle S_t : O_t \sqra f_t$, 
    $R'_t = H'_t : M'_t \sqle S'_t : O'_t \sqra V_t$. 
    There should exist a substitution $\sigma$ 
    such that $\sigma f_t = f_o$, $\sigma H_t \subseteq H_o$, $\sigma M_t \subseteq M_o$, 
    $\sigma S_t \subseteq S_o$, $\sigma O_t \subseteq O_o$, 
    $\sigma f'_t = f'_o$, $\sigma H'_t \subseteq H'_o$, $\sigma M'_t \subseteq M'_o$, 
    and $\sigma S'_t \subseteq S'_o$, $\sigma O'_t \subseteq O'_o$. 
    Assume 
    $S_o \circ_{f} S'_o = 
        \sigma'(H_o \cup (H'_o - \{ f \})) : \sigma'(M_o \cup M'_o) 
        \sqle \sigma'(S_o \cup S'_o) : \sigma'(O_o \cup O'_o \cup S_o \times S) \sqra \sigma' V_o$ 
    where $S = \{ s | \langle f, s \rangle \in M'_o \}$. 
    We discuss the two cases as follows. 
    
    \emph{First case.} Suppose $\exists f' \in H'_t$ such that $\sigma f' = f$. 
    Since $R_o \circ_f R'_o$ is defined and $\sigma' f = \sigma' f_o$, 
    we thus have $\sigma'\sigma f' = \sigma'\sigma f_t$. 
    As $f'$ and $f_t$ are unifiable, $S_t \circ_{f'} S'_t$ is defined. 
    Let $\sigma_t$ be the most general unifier, 
    then $\exists \sigma'_t$ such that $\sigma'\sigma = \sigma'_t\sigma_t$. 
    Suppose we have 
    $S_t \circ_{f'} S'_t = 
        \sigma_t(H_t \cup (H'_t - \{ f' \})) : \sigma_t(M_t \cup M'_t) 
        \sqle \sigma_t(S_t \cup S'_t) : \sigma_t(O_t \cup O'_t \cup S_t \times S') \sqra \sigma_t V_t$ 
    where $S' = \{ s | \langle f', s \rangle \in M'_t \}$. 
    First we prove $\sigma S' = \{ s | \langle \sigma f', s \rangle \in \sigma M'_t \} 
        = \{ s | \langle f, s \rangle \in \sigma M'_t \} 
        \subseteq \{ s | \langle f, s \rangle \in M'_o \} = S$. 
    Since $\sigma'_t\sigma_t(H_t \cap (H'_t - \{ f' \})) 
        = \sigma'\sigma(H_t \cup (H'_t - f')) 
        \subseteq \sigma'(H_o \cup (H'_o - \{ f \}))$, 
    $\sigma'_t\sigma_t(S_t \cup S'_t) = \sigma'(\sigma S_t \cup \sigma S'_t)
        \subseteq \sigma'(S_o \cup S'_o)$, 
    $\sigma'_t\sigma_t(O_t \cup O'_t \cup S_t \times S') 
        = \sigma'(\sigma O_t \cup \sigma O'_t \cup \sigma S_t \times \sigma S')
        \subseteq \sigma'(O_o \cup O'_o \cup S_o \times S)$, 
    $\sigma'_t\sigma_t((M_t \cup M'_t) \cdot (O_t \cup O'_t \cup S_t \times S'))
         \subseteq \sigma' ((M_o \cdot M'_o) \cup (O_o \cup O'_o \cup S_o \times S))$, 
    and $\sigma'_t\sigma_t V_t = \sigma'\sigma V_t = \sigma' V_o$, 
    we have $R_t \circ_{f'} R_t' \Rightarrow R_o \circ_{f} R'_o$. 
    
    \emph{Second case.} $\sigma H'_t \subseteq H'_o - \{ f \}$, 
    then $\sigma' \sigma H'_t \subseteq \sigma' (H_o \cup (H'_o - \{ f \}))$, 
    $\sigma' \sigma (M'_t \cdot O'_t) \subseteq \sigma' (M'_o \cdot O'_o) 
        \subseteq \sigma' (M_o \cdot O_o \cup M'_o \cdot O'_o)$, 
    $\sigma' \sigma S'_t \subseteq \sigma' S'_o \subseteq \sigma' (S_o \cup S'_o)$, 
    $\sigma' \sigma O'_t \subseteq \sigma' O'_o \subseteq \sigma' (O_o \cup O'_o \cup S_o \times S)$, 
    and $\sigma' \sigma V_t = \sigma' V_o$. 
    Therefore $R_t' \Rightarrow R_o \circ_{f} R'_o$. 
\end{proof}


\begin{theorem}\label{thm:base} 
    Any event $e$ that is derivable from the initial rules $\mathcal{B}_{init}$ 
    if and only if it is derivable from the knowledge base $\mathcal{B}$ constructed in Algorithm~\ref{alg:base}.
\end{theorem}

\begin{proof}
    \textbf{Only if.} 
    Assume the event $e$ is derivable from $\mathcal{B}_{init}$, 
    then there should exist a derivation tree $T_i$ for $e$ 
    and every node in the tree is labeled by a rule in $\mathcal{B}_{init}$. 
    According to the $\mathit{add}$ function in Algorithm~\ref{alg:base}, 
    a rule is removed only if it is implied by another rule, 
    so we have $\forall R \in \mathcal{B}_{init}, \exists R' \in \mathcal{B}_v, R' \Rightarrow R$, 
    where $\mathcal{B}_v$ appears at the line~\ref{alg:base:bv} in Algorithm~\ref{alg:base}. 
    Hence, we can replace all the rules labeled on tree with the rules in $\mathcal{B}_v$
    and get a new derivation tree $T_v$. 
    As can be seen from Algorithm~\ref{alg:base}, 
    some rules are filtered out from $\mathcal{B}_v$ to $\mathcal{B}$, 
    so we need to further prove that the nodes in $T_v$ can be composed and transformed 
    until a derivation tree $T$ is formed such that all the rules labeled on $T$ are rules in $\mathcal{B}$. 
    
    To continue the proof, we consider $T_v$ purely as a tree structure, 
    and each tree consists of a root and several connected sub-trees. 
    Next, we prove that each sub-tree is implied by a state consistent rule in $\mathcal{B}_v$. 
    Since the leaves of $T_v$ are implied by the state consistent rules, 
    the sub-trees of the leaves are directly implied by rules in $\mathcal{B}_v$. 
    Given two nodes $n$ and $n'$, $n$'s outgoing edge $f$ is one of incoming edges of $n'$. 
    Assume the subtree $n$ is implied by a state consistent rule $R$ in $\mathcal{B}_v$, 
    the node $n'$ is labeled by a rule $R'$ and $n'$ has a outgoing edge of $f'$. 
    \begin{itemize}
        \item  
        If $f \neq f'$, 
        we have $R \Rightarrow H : M \sqle S : O \sqra f$, 
        $R' \Rightarrow H' : M' \sqle S' : O' \sqra V$ and $f \in H'$. 
        Since $R_f = (H : M \sqle S : O \sqra f) \circ_f (H' : M' \sqle S' : O' \sqra V)$ is defined, 
        according to Lemma~\ref{lem:combine}, 
        the sub-tree $n'$ is also implied by a rule in $\mathcal{B}_v$ in two cases. 
        In the first case, there exists $f''$ in the predicates of $R'$, $R \circ_{f''} R' \Rightarrow R_f$. 
        If $f''$ is not a singleton, 
        because $\mathcal{B}_v$ is the fixed-point of Algorithm~\ref{alg:base}, 
        there should exist $R'' \in \mathcal{B}_v$ such that $R'' \Rightarrow R_f$. 
        So we can merge these two nodes in the tree and the proof continues. 
        Otherwise, i.e., $f''$ is a singleton, 
        we can detach the sub-tree of $n$ from tree $T_v$ temporarily. 
        With the composition and transformation processing, 
        $f''$ may be unified to a non-singleton fact, so the composition could continue. 
        If the other part of the tree has been processed and $f''$ is still a singleton, 
        we will prove later that $n$ can be removed from the tree and the derivation tree is still valid. 
        In the second case, we can remove the node $n$ and link its incoming links directly to $n'$, 
        so that the node $n'$ with more incoming edges is still implied by $R'$ and the proof continues. 
        \item 
        If $f = f'$, apparently we have that $R$ implies the subtree of $n'$. 
    \end{itemize}
    We can continue the rule composition until we reach the root 
    so that each subtree in $T_v$ is implied by a state consistent rule in $B_v$. 
    
    Notice that the states are not properly transferred in the rule that is labeled to the tree $T_v$, 
    so we also need to re-organize the states in the rule to form a valid state trace. 
    Since all the state duplications appear in the sub-tree are kept in the resulting rule, 
    we will merge them according to the state transformation. 
    Consider the root is labeled by a rule $R_r$, 
    all the states appeared in the tree $T_v$ should be presented in $R_r$. 
    According to the derivation tree, some of the edges are labeled by the same index. 
    So we prove in the following iterations, the resulting rule is still in $\mathcal{B}_v$. 
    The index starts with $1$, which is same index of the root, 
    and it is increased by $1$ after every iteration. 
    If currently the index is $i$, since the states in the rule are corresponding to the states in the edge, 
    so we can merge the states in the edges labeled by $i$ together. 
    According to the definition of the derivation tree, 
    from the edges labeled by $i+1$ to the edges labeled by $i$, 
    there exists a conversion set $C$ that converts some old states to the new states. 
    Hence, we could construct the mapping function $m$ defined in the state transformation, 
    and map each $c \in C$ to a set of states that should be merged 
    (the latest states for the latest state transferring rule). 
    After the state transformation, the largest states in the rule now are labeled by index $i+1$. 
    According to Algorithm~\ref{alg:base} case 2, 
    the new rule should be also in $\mathcal{B}_v$. 
    Notice that we have mentioned previously that some rules cannot be composed because 
    the incoming edge of the rule is labeled by a singleton. 
    Along with the state transformations, some singleton may be unified to a non-singleton fact, 
    so the rule composition could continue. 
    
    In this way, the rule composition and the state transformation can be conducted 
    until all states left are all labeled by the largest index. 
    If some inner edges are still labeled by singletons, 
    because the adversary can generate new names, 
    he can actively create a new value and label it to that edge, 
    so that he can drop the remaining sub-tree connected by that edge 
    and the remaining derivation tree is still valid. 
    Since the facts in leaves are the events and singletons, 
    including those failed with unification, 
    the resulting rule is in the output knowledge base $\mathcal{B}$. 
    
    \textbf{If.}
    Whenever a rule is added into $\mathcal{B}_v$, it should be composed or transferred from existing rules. 
    Thus all the rules in $\mathcal{B}_v$ should be derivable from $\mathcal{B}_{init}$. 
    Meanwhile $\mathcal{B}$ does not introduce extra rules 
    besides existing rules in $\mathcal{B}_v$, 
    so $\forall R' \in \mathcal{B}$, $R'$ is derivable from $\mathcal{B}$. 
\end{proof}


\subsection{Reachability Analysis} 
\label{sub:searching}

\begin{algorithm}[t]
\small
    \SetAlgoLined
    \SetKwInOut{Input}{Input}
    \SetKwInOut{Output}{Output}
    \Input{$\mathcal{B}$ - the knowledge base returned by Algorithm~\ref{alg:base}. }
    \Input{$e$ - the target event. }
    \Output{$b$ - if the event is reachable or not. }
    \SetKwProg{algo}{Algorithm}{}{}
    \algo{}{
        \For{$f_1, \ldots, f_n : M \sqle S : O \sqra f \in \mathcal{B}$ and $f = e$}{
            \lIf{$\exists \sigma, \forall s, s' \in S, s \sim s' \Rightarrow$ \\
            $\sigma s = \sigma s'$ and $\sigma s$ is accessible}
            {\Return $\mathit{true}$}
        }
        \Return $\mathit{false}$; \\
    }
    \caption{Query Contradiction Searching}
    \label{alg:search}
\end{algorithm}

When the knowledge base is constructed, 
we need to check if the target event is reachable or not. 
Given a rule in the base, 
if the predicates are only events and singletons, 
the adversary can fulfill them by asking the protocol to engage those events and generate new names. 
For the remaining states in the rule, 
we then need to check if the adversary has the access to the corresponding object patterns. 
Assume the accessibility is modeled as a set of state patterns $P$ according to Section~\ref{sub:accessibility}. 
We define a state $s$ as accessible to the adversary if 
$\exists p \in P$ such that $\exists \sigma$, $\sigma s = p$. 
For instance, if the attack needs a TPM $tpm(cary[], p)$ from another participant Cary, 
while \emph{Bob} only have the access to the TPM from himself and not preciously owned TPMs. 
Since there does not exist such a substitution $\sigma$ such that 
$\sigma tpm(cary[], p) = tpm(bob[], p')$ or $\sigma tpm(cary[], p) = tpm(aik, p')$, 
the attack found is impractical. 
Thus, a query can be answered using a simple algorithm as shown in Algorithm~\ref{alg:search}. 
It checks the target event against all the remaining rules in the knowledge base $\mathcal{B}$, 
and tries to find a rule whose predicates can be fulfilled and states can be accessed by the adversary. 
If there exists such a rule, the algorithm returns true; 
otherwise it returns false. 
We prove the partial correctness of our algorithm as follows. 
\begin{theorem}\label{thm:search}
    An event $e$ is derivable from the initial rules $\mathcal{B}_{init}$ if and only if 
    there exists a rule in $\mathcal{B}$ such that its conclusion is $e$ 
    and its states are all accessible to the adversary.
\end{theorem}
\begin{proof}
    \textbf{(If - Partial Soundness)}
    If there is a rule in $\mathcal{B}$ that outputs $e$. 
    As the rules' predicates are events and singletons, 
    the adversary can ask the protocol to engage those events and generate new names to fulfill the singletons. 
    When its states in the same partition are unifiable and all unified states are accessible to the adversary, 
    the adversary can have the objects to meet the requirements of those states. 
    Hence, $e$ is derivable by the rule. 
    According to Theorem~\ref{alg:base}'s if condition, 
    $e$ is also derivable from $\mathcal{B}_{init}$. 
    
    \textbf{(Only if - Partial Completeness)}
    If the event $e$ is derivable from $\mathcal{B}_{init}$, 
    according to Theorem~\ref{thm:base}'s only if condition, $e$ is also derivable from $\mathcal{B}$. 
    As the derivation tree is valid, 
    the initial states should be accessible states for the adversary. 
\end{proof}



\section{Experiments} 
\label{sec:experiments}

Our engineering efforts has realized the proposed approach in a tool named SSPA (Stateful Security Protocol Analyzer). 
Our tool, all protocol models and evaluation results are available online at~\cite{LPD14}. 
SSPA is implemented in C++ with around 11K LOC. 
The experiments presented in this section are evaluated 
with Mac OS X 10.9.1, 2.3 GHz Intel Core i5 and 16G 1333MHz DDR3. 

We have tested our tool with three versions of the DEP~\cite{AR10,DKRS11}, 
the Bitlocker protocol~\cite{BL11} and two versions of the Needham-Schroeder Public Key Protocol (NSPK)~\cite{NS78,Low95}. 
All of the protocols are correctly analyzed within 30 minutes. 
The results are summarized in Table~\ref{tab:experiments}. 

\begin{table}[t]
    \begin{center}
    \begin{tabular}{l | c l r}
        \hline \hline
        ~Protocol   & $\sharp$Rules\footnote{The number of rules generated 
            by our solving algorithm for each protocol.} & Result & Time \\
        \hline
        DEP (w.o. reboot)~\cite{AR10}     & 318       & Secure    & 6.2s \\
        DEP (w. reboot)~\cite{AR10}      & 1409      & Attack    & 12m 9.5s \\
        Modified DEP~\cite{DKRS11}      & 1378      & Secure    & 22m 17.7s \\
        \hline
        Bitlocker~\cite{BL11}          & 24        & Secure    & 3ms \\
        \hline
        NSPK~\cite{NS78}               & 101       & Attack    & 47ms \\
        NSPK (Lowe)~\cite{Low95}        & 78        & Secure    & 24ms \\
        \hline
    \end{tabular}
    \end{center}
    \caption{Experiment results}
    \label{tab:experiments}
\end{table}

For the DEP example, when the TPM reboot is disallowed, the verification result shows that 
\emph{Bob} cannot obtain both of the secret and the proof for his revoke action at the same time. 
In the meanwhile, we also found several valid traces for \emph{Bob} to finish the protocol 
by either opening the envelope or revoking his right. 
However, when the TPM reboot is allowed, 
the claimed security property of the DEP is not preserved. 
In addition to the attack trace described in Section~\ref{sec:example}, 
SSPA also found several other traces (attacking at different states), 
which are similar variants to the attack described in Section~\ref{sec:example}. 
The modified version of the DEP presented in~\cite{DKRS11} is also proven to be secure in our framework. 

The Bitlocker~\cite{BL11} designed by Microsoft also uses TPM to protect its execution state. 
In the machine equipped with Bitlocker, 
the hard drive is assumed to be encrypted under a volume encryption key (VEK). 
The VEK is in turn encrypted by a volume master key (VMK). 
When the machine is booted, 
an immutable pre-BIOS will load the BIOS and extend the hash value of the BIOS into the TPM. 
The pre-BIOS then passes the control to the BIOS. 
Later, the BIOS can load other components by first extending the hash value of that component into the TPM. 
The components then could in turn load other components by doing this repeatedly, resulting in a trust chain. 
Initially, the VMK is sealed by the TPM to a certain PCR value 
corresponding to a correct boot state of the machine. 
When the correct state is reached, the VMK can be unsealed to decrypt the hard drive and access its data. 
Even though the attacker could replace the BIOS and other components in the machine, 
their hash values will not be the same as the original ones. 
So the correct state cannot be reached and the VMK remains secure. 
We model the protocol by assuming that the attacker can read the VMK 
by either replacing a fake BIOS or a fake loader (a component) in the machine. 
Otherwise, the attacker cannot access the unsealed data from the machine even if it is unsealed 
as it is controlled by a trusted component. 
The verification result shows that Bitlocker protects the VMK from the attacker 
even when the BIOS and the loader can be replaced. 

Lastly, we modeled the Needham-Schroeder Public Key (NSPK) Protocol~\cite{NS78} 
and its fixed version by Gavin Lowe~\cite{Low95}. 
We use these two examples to show that our approach also works for stateless protocols. 
In order to model the nonces exchanged by the participants in NSPK as random numbers, 
we add two states for the participants when their first message is sent 
and they are waiting for the second message by treating them as trusted parties. 


\section{Discussions} 
\label{sec:discussions}

In this paper, we have presented a new approach
for the stateful security protocol verification.
Different from existing tools in the literature,
our approach allows for specifying stateful protocols directly (without modifications to the protocols)
and it can deal with infinite protocol states. 
Moreover, our verification procedure is sound and complete if the solving algorithm terminates.
We have implemented a tool for our new approach
and validated it on a number of protocols. 
So far, the initial results are encouraging. 

When rules are newly composed in the knowledge base, 
the redundancy checking consumes a large amount of time. 
This is mainly because of the complexity of pairing states and predicates from different rules 
and finding all possible substitutions according to Definition~\ref{def:implication}. 
For the future work, accelerating the redundancy checking would be very helpful 
to accelerate the verification process dramatically. 
In addition, analyzing more stateful protocols would be very interesting. 
Moreover, adapting our approach to verify stateful protocols 
with physical properties involved, e.g., time, space, etc. would be promising as well. 


\bibliographystyle{IEEEtran}
\bibliography{paper}

\end{document}